\documentclass{math-ppn}

\usepackage[T1]{fontenc}
\usepackage[utf8]{inputenc}
\usepackage{amsmath,amsfonts,amsthm,dsfont,amssymb,amsxtra,bbm}
\usepackage{setspace,graphicx,color}
\usepackage[colorlinks=true]{hyperref}
\hypersetup{urlcolor=blue, linkcolor=blue, citecolor=red, anchorcolor=blue]}

\newtheorem{theorem}{Theorem}

\newtheorem{lemma}[theorem]{Lemma}

\newcommand{\Proba}{{\mathbb P}}
\newcommand{\E}{{\mathbb E}}

\newcommand{\be}[1]{\begin{equation}\label{#1}}
\newcommand{\ee}{\end{equation}}
\renewcommand{\(}{\left(}
\renewcommand{\)}{\right)}

\newcommand{\s}{\mathsf s}
\newcommand{\e}{\mathsf e}
\renewcommand{\i}{\mathsf i}
\renewcommand{\r}{\mathsf r}

\newcommand{\Email}[2]{E-mail~(#1):~\href{mailto:#2}{\texttt{#2}}}

\begin{document}

\title{Heterogeneous social interactions and the COVID-19 lockdown outcome in a multi-group SEIR model}

\author{Jean Dolbeault}
\address{CEREMADE (CNRS UMR n$^\circ$ 7534), PSL university, Universit\'e Paris-Dauphine, Place de Lattre de Tassigny, 75775 Paris~16, France. \Email{J.D}{dolbeaul@ceremade.dauphine.fr}. \Email{G.T.}{turinici@ceremade.dauphine.fr}.\medskip}

\author{Gabriel Turinici$^1$}

\begin{abstract} We study variants of the SEIR model for interpreting some qualitative features of the statistics of the Covid-19 epidemic in France. Standard SEIR models distinguish essentially two regimes: either the disease is controlled and the number of infected people rapidly decreases, or the disease spreads and contaminates a significant fraction of the population until herd immunity is achieved. After lockdown, at first sight it seems that social distancing is not enough to control the outbreak. We discuss here a possible explanation, namely that the lockdown is creating social heterogeneity: even if a large majority of the population complies with the lockdown rules, a small fraction of the population still has to maintain a normal or high level of social interactions, such as health workers, providers of essential services, \emph{etc}. This results in an apparent high level of epidemic propagation as measured through re-estimations of the basic reproduction ratio. However, these measures are limited to averages, while variance inside the population plays an essential role on the peak and the size of the epidemic outbreak and tends to lower these two indicators. We provide theoretical and numerical results to sustain such a view. \end{abstract}

\begin{resume} Nous étudions des variantes du modèle SEIR afin d'interpréter certaines caractéristiques qualitatives des statistiques de l'épidémie de Covid-19 en France. Les modèles SEIR standards distinguent deux régimes: soit la maladie est contrôlée et le nombre de personnes infectées diminue rapidement, soit la maladie se propage et contamine une fraction importante de la population jusqu'à ce que l'immunité collective soit atteinte. Après le confinement, il semble à première vue que la distanciation sociale soit insuffisante pour contrôler l'épidémie. Nous avançons ici une explication possible, à savoir que le confinement crée de l'hétérogénéité sociale: même si une grande majorité de la population obéit aux règles de confinement, une petite fraction doit continuer à maintenir un niveau normal ou élevé d'interactions sociales, comme les personnels médicaux, les prestataires de services essentiels, \emph{etc}. Cela se traduit par un niveau de propagation élevé de l'épidémie, mesuré par des ré-estimations du taux de reproduction de base. Ces mesures se limitent toutefois à des moyennes alors que la variance au sein de la population joue un rôle essentiel sur le pic et la taille de l'épidémie et contribue à abaisser ces deux indicateurs. Nous apportons des arguments théoriques et numériques pour développer ce point de vue. \end{resume}

%
%


\date{\today}

\maketitle

\section*{Introduction}

Although widely used in practice, compartmental models of epidemic spread (as for instance the celebrated SIR\cite{zbMATH02582630} model) rely on various simplifying assumptions, like a limited number of compartments, the homogeneity of the population inside a compartment, or some well defined parameters. The advantage of such models is that the impact of the variation of a single parameter on various qualitative properties can easily be studied, but the risk lies in an outrageous simplification of the representation of a complex system, with the additional difficulty that some parameters can be difficult to quantify from the available statistical data. Here we numerically study how the separation of the population into two sub-groups of individuals with different intensities of social interaction can qualitatively explain some observed features of the current pandemic of Covid-19, and provide some theoretical explanations which also apply to more realistic models.

New characteristics of the pandemic of Covid-19 are unveiled every day and reveal various interesting features. One of the issues is that curves showing the number of new cases in European countries stabilize very slowly after the beginning of the lockdown. Many questions have been raised on the methods for collecting the data and on their quality, that we do not address, not to mention variations between countries that deserve further investigations. As a simple explanation, it is has been suggested that the main reason is that a significant fraction of the population does not respect lockdown. This ignores the nonlinear properties of simple epidemic models. We show here that, if a rather small fraction of the population cannot reduce its social interactions, eventually for very good reasons as, \emph{e.g.}, the health workers or some other key actors of our societies, the epidemic keeps spreading, stabilizes at a much slower rate than what one would expect, and finally affects a significant fraction of the population.

In this paper we shall use the SEIR model rather than the simpler SIR model because of the significant period of incubation in the Covid-19 disease. The SEIR model has various properties that can be understood, at least for the order of magnitude of its effects, using simple objects like the well known \emph{basic reproduction ratio} $\mathcal R_0$. However, it also has very nonlinear features which are defying common sense and require rigorous mathematical analysis. In the accompanying numerical examples, we take values for the parameters which are compatible with the data collected during the pandemic of Covid-19. These examples are intended to understand qualitative features of the epidemic but cannot and should not be implemented for direct predictive use, as our model is too crude and oversimplified to reflect the Covid-19 epidemic in a quantitative manner.

We start by reviewing known results concerning the SEIR model applied to a homogeneous population in Section~\ref{Sec:SEIR1G}. A model with a large majority of the population under lockdown and a small minority which does not or cannot implement the \emph{social distancing} is then considered. The \emph{factor of reduction of social interactions}~$q$ quantifies the effect of the lockdown. It is certainly very difficult to measure this parameter in real life applications and structured population models should be considered, in order to reflect the variety of social interactions. Our goal is to understand how, by varying $q$, one eventually triggers the nonlinearity of the SEIR model and, as a consequence, how this drives the system into a \emph{stable equilibrium} which is far from the usual target of a lockdown policy, that is, the control of the disease, but is also far from the dynamics that would develop in absence of lockdown.

\section{Basic mathematical properties of the SEIR model}\label{Sec:SEIR1G}

\subsection{The SEIR model}

Let us consider the SEIR (\emph{Susceptible}, \emph{Exposed}, \emph{Infected}, \emph{Recovered}) model~\cite{anderson1992infectious} defined by the set of equations
\begin{align}
&\frac{dS}{dt}=-\,\beta\,S\,\frac IN\,,\label{S}\\
&\frac{dE}{dt}=\beta\,S\,\frac IN-\,\alpha\,E\,,\label{E}\\
&\frac{dI}{dt}=\alpha\,E-\,\gamma\,I\,,\label{I}\\
&\frac{dR}{dt}=\gamma\,I\,,\label{R}
\end{align}
which is a variant of the SIR model~\cite{zbMATH02582630} of Kermack and McKendrick. Here we neglect birth and death rates, \emph{i.e.}, we consider a model without vital dynamics. The \emph{average incubation period} is $1/\alpha$, the parameter $\beta$ is the product of the average number of contacts per person and per unit time by the probability of disease transmission in a contact between a susceptible and an infectious individual, $\gamma$ is a \emph{transition rate} so that $1/\gamma$ measures the duration of the infection of an individual and $N$ is the total population size. In the Covid-19 pandemy, the average incubation period is of several days and this is why an SEIR model has to be preferred to a simple SIR model. Many qualitative features are the same in the two models but the compartment $E$ of exposed individuals makes the analysis slightly more delicate. Unreported cases or asymptomatic individuals are not taken into account here: this is an important aspect of the Covid-19 epidemic, see for instance~\cite{danchin_new_2020,LiuMagal_2020,magal2020.03.21.20040154}, with important consequences on the epidemic size, but probably not so much on the qualitative issues that are discussed in this paper. Other factors, like delays for the transmission of the information studied in~\cite{Buonomo_2008}, certainly also play a role in the current outbreak.

The SEIR model is a \emph{compartmental model} used to understand the mathematical modelling of infectious diseases in a large population, with enough individuals in each compartment so that stochastic effects can be neglected. Spatial effects are also neglected, which is a rather crude approximation. However, such a simplified model allows us to perform an analysis of the sensibility to the parameters which are of interest for a qualitative description of the outbreak of an epidemic disease.

The system~\eqref{S}-\eqref{R} is homogeneous so that we can simply consider the \emph{fractions}
\[
\s:=\frac SN\,,\quad\e:=\frac EN\,,\quad\i:=\frac IN\,,\quad\r:=\frac RN
\]
of the \emph{Susceptible}, \emph{Exposed}, \emph{Infected} and \emph{Recovered} individuals among the whole population.

\subsection{Conservations and large time asymptotics}

Here we perform a simple analysis of the model, which is done in the spirit of~\cite{MR1814049}. By conservation of the total number of individuals, we have the relation
\be{1}
\s(t)+\e(t)+\i(t)+\r(t)=1
\ee
for any time $t\ge0$, which is easily proved by summing the equations in the system, while the evolution is now governed by the system
\begin{align}
&\frac{d\s}{dt}=-\,\beta\,\s\,\i\,,\label{s}\\
&\frac{d\e}{dt}=\beta\,\s\,\i-\,\alpha\,\e\,,\label{e}\\
&\frac{d\i}{dt}=\alpha\,\e-\,\gamma\,\i\,,\label{i}\\
&\frac{d\r}{dt}=\gamma\,\i\,.\label{r}
\end{align}
An important and classical observation is that the domain
\[
\s\ge0\,,\quad\e\ge0\,,\quad\i\ge0\,,\quad\s+\e+\i\le1
\]
is stable under the action of the flow~\eqref{s}-\eqref{r}, and it is straightforward to check that
\be{si}
\frac d{dt}\(\gamma\,\log\s-\beta\,(\s+\e+\i)\)=0\,.
\ee
Since $\s$ is monotone non-increasing by~\eqref{s}, this means that~\eqref{s}-\eqref{r} has a global solution for any $t\ge0$ and
\[
\lim_{t\to+\infty}\s(t)=\s^\star>0\,.
\]
By an elementary analysis, we find that any solution $(\s,\e,\i,\r)$ of~\eqref{s}-\eqref{r} with initial data $(\s_0,\e_0,\i_0,\r_0)$ such that $\i_0+\e_0>0$ converges as $t\to+\infty$ to a stationary solution $(\s^\star,\e^\star,\i^\star,\r^\star)$ with
\be{star}
\e^\star=\i^\star=0\,,\quad\s^\star+\r^\star=1\quad\mbox{and}\quad\gamma\,\log\s^\star-\beta\,\s^\star=\gamma\,\log\s_0-\beta\,(\s_0+\e_0+\i_0)\,,
\ee
according to~\eqref{1} and~\eqref{si}. Note that the solution is unique if $\gamma-\beta>0$ but there are two solutions if $\gamma-\beta<0$ and $\gamma\,\log\s_0+\beta\,\r_0=\gamma\,\log\s_0-\beta\,(\s_0+\e_0+\i_0)+\,\beta>0$. In our numerical applications, we shall assume that $\gamma\,\log\s_0+\beta\,\r_0=\gamma\,\log\s_0-\beta\,(\s_0+\e_0+\i_0)+\,\beta<0$ if $\gamma-\beta<0$, so that $\s^\star$ is uniquely defined. 

The \emph{epidemic size} $\zeta$ is defined as the fraction of individuals that are affected by the epidemic, here $\s_0-\s^\star$. Next, let us consider more specifically the case of an initial datum which is a perturbation of the constant in time \emph{DFE solution}, or \emph{Disease Free Equilibrium} (see~\cite{van_den_Driessche_2002}), corresponding to
\be{DFE1}
(\s_0,\e_0,\i_0,\r_0)=(1,0,0,0)\,.
\ee

\subsection{Stable equilibrium, epidemic size and phase transition}\label{Sec:Equilibrium1}

Let us assume that $\e_0+\i_0+\r_0=\varepsilon>0$ is small and $\r_0=\vartheta\,\varepsilon$ for some given $\vartheta\in[0,1)$. This initial condition, in the limit as $\varepsilon\to0_+$, is a perturbation of the DFE solution. For any $\varepsilon>0$, we know by~\eqref{1} that $\s_0=1-\varepsilon$. The solution of~\eqref{s}-\eqref{r} converges as $t\to+\infty$ to $(\s^\star,0,0,\r^\star=1-\s^\star)$, the stationary solution $(\s^\star,0,0,\r^\star)$ is stable and this is why we call it the \emph{stable equilibrium} solution. See~\cite[Section~7.3]{MR3409181} for a discussion of the stability. In a model with birth and death rates, the solution is usually called the \emph{endemic equilibrium}, but we shall prefer to call it simply the \emph{stable equilibrium} solution as we neglect birth and death issues. The following discussion is given here in preparation for the next two sections. We do not claim originality and refer for instance to~\cite{MR1814049} for a detailed study motivated by very explicit examples. Depending whether the DFE solution is stable or not, we have three regimes for $\s^\star=\s^\star(\varepsilon)$:
\\[8pt]
$\rhd$ \textbf{\textit{Control of the epidemic.}} If $\gamma>\beta$, we find that
\be{Control}
\s^\star(\varepsilon)=\s_0-\frac{\beta\,\varepsilon}{\gamma-\beta}+o(\varepsilon)=1-\frac{\gamma\,\varepsilon}{\gamma-\beta}+o(\varepsilon)
\ee
as $\varepsilon\to0_+$. In other words, the stable equilibrium is a perturbation of the DFE solution and the \emph{epidemic size}~$\zeta$ is of order $\varepsilon$.
\\[8pt]
$\rhd$ \textbf{\textit{Epidemic spreading and herd immunity.}} If $\gamma<\beta$, then we note that $s\mapsto\gamma\,\log s-\beta\,s$ achieves a maximum point on $(0,1)$ at $s=\gamma/\beta$ and, as a consequence, that
\be{stareps}
\s^\star(\varepsilon)<\frac\gamma\beta
\ee
even for an arbitrarily small value of $\varepsilon>0$. Hence we find that $\s^\star(\varepsilon)$ is not only a perturbation of order $\varepsilon$ of~$\s_0$ but also that the value of $\s^\star(\varepsilon)$ is of order $\s^\star(0)$, the unique root in $(0,1)$ of $\gamma\,\log s+\beta\,(1-s)=0$. Since the constant solution $\big(\s^\star(\varepsilon),0,0,1-\s^\star(\varepsilon)\big)$ is a stable equilibrium, \emph{herd immunity} is always granted in the sense that no outbreak can occur. Let $\mathcal R_0=\beta/\gamma$.

With $\mathcal R_0>1$, the epidemic spreads in the sense that $\s^\star(0)$ is the solution~of
\be{stareq}
\log s-\mathcal R_0\,(s-1)=0\,.
\ee
This means that $\s^\star(0)=-\frac1{\mathcal R_0}\,W\(-\mathcal R_0\,e^{-\mathcal R_0}\)$ where the function $W$ is known as the \emph{Lambert function} (see for instance~\cite{Corless_1996}) and defined as the inverse of $w\mapsto w\,e^w$. Note that $r=1-s$ is given for $s=\s^\star(0)$ by the equation
\be{req}
e^{-\mathcal R_0\,r}+r=e^{-\frac{\beta\,r}\gamma}+r=1
\ee
and that $r=1-\s^\star(\varepsilon)$ solves
\be{reqeps}
\s_0\,e^{-\frac\beta\gamma\,(r-\r_0)}+r=1\,,\quad\s_0=1-\varepsilon\,,
\ee
so that $\s^\star(\varepsilon)$ also depends on $\vartheta$ (but this dependence disappears in the limit as $\varepsilon\to0_+$). Altogether, if $\gamma<\beta$, that is, for $\mathcal R_0>1$, the \emph{epidemic size} $\zeta$ is of order $1$ as $\varepsilon\to0_+$. See Fig.~\ref{f3} for an illustration of the dependence of $\zeta$ in $\mathcal R_0$.
\\[8pt]
$\rhd$ In the \textbf{\textit{threshold case}} $\beta=\gamma$, \emph{i.e.}, $\mathcal R_0=1$, which is typical of a phase transition, we have to solve
\be{sstardev}
\log\s^\star(\varepsilon)-\s^\star(\varepsilon)=\log(1-\varepsilon)+\r_0-1
\ee
and find that
\be{sstardev2}
\s^\star(\varepsilon)=1-\sqrt{2\,(1-\vartheta)\,\varepsilon}+o\big(\sqrt\varepsilon\big)
\ee
as $\varepsilon\to0_+$. We recall that $0\le\r_0<\varepsilon$ and observe that the \emph{epidemic size} $\zeta$ is of order $\sqrt\varepsilon$ as $\varepsilon\to0_+$. Scale invariance is reflected by the fact that there is no dependence neither on $\beta$ nor on $\gamma$.

Summarizing, there are two phases and a threshold case corresponding to the \emph{phase transition}. If $\gamma>\beta$, the epidemic size $\zeta$ is close to zero. The disease does not spread in the population and simply vanishes exponentially fast. On the opposite, if $\gamma<\beta$, the \emph{Disease Free Equilibrium} is unstable, the diseases quickly spreads with an exponential growth, the system converges for large times to a stable equilibrium far away from the DFE solution and the epidemic size $\zeta$ is a significant fraction of the total population, irrespective of how small the fraction of initially infected individuals~is. Whether $\mathcal R_0>1$ or $\mathcal R_0<1$ determines the asymptotic stable equilibrium starting from the DFE solution. In the literature, $\mathcal R_0$ is called the \emph{basic reproduction ratio}, or the \emph{basic reproduction number}. So far we simply consider it here as the \emph{order parameter} of the phase transition, in the usual sense in physics: see for instance~\cite[p.~449]{landau2013statistical} or~\cite[p.~3]{stanley1971introduction}. In the next section, we shall explain the role it plays in the initial dynamics of the model and refer to Section~\ref{sec:roliterature} for more considerations on the epidemiologic interpretation of $\mathcal R_0$.

Note that the SIR model has exactly the same stationary states and the same phase transition as the SEIR model, as the SIR model is obtained by replacing~\eqref{e} by $\beta\,\s\,\i=\alpha\,\e$, so that, in the SIR model, the equation for~$\i$ becomes
\be{sir}
\frac{d\i}{dt}=\beta\,\s\,\i-\,\gamma\,\i\,.
\ee
The order parameter is also $\mathcal R_0=\beta/\gamma$ and whether $\mathcal R_0>1$ or $\mathcal R_0<1$ determines if the epidemic is spreading or if the disease is controlled.

\subsection{Linearization and the basic reproduction ratio}

In the case of the SIR model,~\eqref{sir} can be rewritten as
\be{sir2}
\frac{d\i}{dt}=\(\beta\,s-\,\gamma\)\i=\big(\mathcal R_0\,\s-1\big)\,\gamma\,\i
\ee
and it is elementary to observe that when $\s\sim1$, whether $\mathcal R_0>1$ or $\mathcal R_0<1$ determines the initial dynamics of the model: the interpretation of $\mathcal R_0$ is clear from the above equation. After this digression on the SIR model, let us come back to the SEIR model. Understanding the role of $\mathcal R_0$ is a little bit more subtle than in the SIR model. We recall that $\r(t)=1-\s(t)-\e(t)-\i(t)$ plays no role in the stability analysis of the DFE solution. At any time $t$, the linearized dynamics of $t\mapsto\big(\s(t),\e(t),\i(t)\big)$ 
is described by the matrix
\begin{equation}
\mathcal M(\s,\i):=\(\begin{array}{ccc}
-\,\beta\,\i&0&-\,\beta\,\s\\
\beta\,\i&-\,\alpha&\beta\,\s\\
0&\alpha&-\,\gamma
\end{array}\)
\end{equation}
and we may notice that the largest eigenvalue of $\mathcal M(1,0)$, corresponding to the linearization around the DFE solution, is
\begin{equation}
\lambda(\alpha,\beta,\gamma):=\frac12\(\sqrt{(\alpha-\gamma)^2+4\,\alpha\,\beta}-\alpha-\gamma\),
\end{equation}
so that $\lambda(\alpha,\beta,\gamma)$ is positive if and only if $\mathcal R_0=\beta/\gamma>1$. Moreover, the eigenspace is compatible with the nonlinear dynamics so that there are perturbations of the DFE solution which are exponentially growing with a rate $\lambda(\alpha,\beta,\gamma)$ if and only if $\mathcal R_0>1$. However, the basic reproduction ratio $\mathcal R_0$ is not anymore directly connected with the linearized growth mode of $\i$. In fact, we can observe that~\eqref{sir2} is replaced by
\be{R0e+i}
\frac d{dt}(\e+\i)=\(\beta\,s-\,\gamma\)\i\sim\big(\mathcal R_0-1\big)\,\gamma\,\i
\ee
if $\s\sim1$, which is indeed the correct way of estimating the growth of the epidemic. We will come back on the interpretation of $\mathcal R_0$ in Section~\ref{sec:roliterature} and note that $(\e+\i)$ is known in the literature as the population in the \emph{infectious compartments}.

\subsection{Social distancing and the factor of reduction of social interactions}

The goal of a lockdown policy is to replace the system~\eqref{s}-\eqref{r} by
\begin{align}
&\frac{d\s}{dt}=-\,\frac\beta q\,\s\,\i\label{sq}\\
&\frac{d\e}{dt}=\frac\beta q\,\s\,\i-\,\alpha\,\e\label{eq}\\
&\frac{d\i}{dt}=\alpha\,\e-\,\gamma\,\i\label{iq}\\
&\frac{d\r}{dt}=\gamma\,\i\label{rq}
\end{align}
for some factor $q>1$ which measures the \emph{reduction of social interactions} of each individual. Of course, what we obtain is exactly~\eqref{s}-\eqref{r} with the parameter $\beta$ replaced by $\beta/q$. The point is that the basic reproduction ratio becomes
\be{R01}
\mathcal R_0^{(1)}(q)=\frac\beta{\gamma\,q}=\frac{\mathcal R_0}q
\ee
and the goal is either to fix $q$ to a value large enough so that the epidemic is controlled, that is $q>\mathcal R_0$, or at least to make $\mathcal R_0-q>0$ small in order to \emph{flatten the curve}, \emph{i.e.}, to have an epidemic going at slower pace. See~Fig.~\ref{f1-f2} for an illustration. Here the exponent $(1)$ in the notation $\mathcal R_0^{(1)}$ points to the assumption that we consider a population with a single group of susceptible individuals or, in other words, a socially homogeneous population.

In a SIR model with social distancing, with~\eqref{sq},~\eqref{rq} on the one hand, but~\eqref{eq} and~\eqref{iq} replaced by
\be{seirq}
\frac{d\i}{dt}=\frac\beta q\,\s\,\i-\,\gamma\,\i
\ee
on the other hand, it is possible to compute the epidemic peak. This is a classical result, see for instance~\cite[Section~2.1.2]{MR3409181}. The epidemic peak is defined as the maximum $\i(t_p)$ of $t\mapsto\i(t)$ after noticing that, at $t=t_p$, we have the system of equations
\be{sirpeak}
\frac\beta q\,\s\,\i-\,\gamma\,\i=0\,,\quad\gamma\,\log\(\frac\s{\s_0}\)+\frac\beta q\,(\r-\r_0)=0\,,\quad\s+\i+\r=1\,,
\ee
which provides us with the value
\be{imax-sir}
\i(t_p)=1-\r_0-\frac1{\mathcal R_0^{(1)}(q)}\,\Big(1+\log\(\mathcal R_0^{(1)}(q)\,\s_0\)\Big)
\ee
for any $q<1/\mathcal R_0$. Note that $\i(t_p)\sim1-\mathcal R_0^{(1)}(q)^{-1}\,\big(1+\log(\mathcal R_0^{(1)}(q))\big)$ as $(\s_0,\r_0)\to(1,0)$, \emph{i.e.}, in the limit of a DFE initial datum. We obtain the same expression in the SEIR model if we replace $\i$ by $\e+\i$. See~\cite{feng2007final} as a source of inspiration for such considerations, and also Property~(2) of Theorem~\ref{Thm:Main}, in a much more general framework.

\subsection{The basic reproduction ratio and the method of the next generation matrix} \label{sec:roliterature}

According to~\cite{MR1057044}, the basic reproduction ratio $\mathcal R_0$ is the \emph{expected number of secondary cases produced, in a completely susceptible population, by a typical infected individual during its entire period of infectiousness}. In~\cite{Blackwood_2018}, Blackwood and Childs provide us with a comprehensive introduction to the computation of $\mathcal R_0$ using the method of the \emph{next generation matrix} in the case of the SEIR model. Such a computation goes back to~\cite{MR1057044}, the standard method is exposed in~\cite{van_den_Driessche_2002} and we can also refer to~\cite{Diekmann_2009} for an application to the SEIR model. As for a more general presentation of the method in compartmental models and more formal mathematical treatments, one can refer to~\cite{diekmann2000mathematical,MR3409181}, and to~\cite{hethcote1987epidemiological} and~\cite[p.~17]{anderson1992infectious} for early considerations on endemic and stable equilibria. Alternative definitions of the basic reproduction ratio in compartmental models are also available: see~\cite{MR3409181} for an overview,~\cite{MR1814049} for a more historical account and~\cite{hethcote1987epidemiological,zbMATH00898403} for considerations on models which are more directly linked to our interests (see Section~\ref{Sec:SEIR2G}).

For sake of completeness, let us give a brief summary of the method of the \emph{next generation matrix}. First of all, one restricts the analysis to the \emph{infectious compartments}, $\mathbf x = (\e,\i)$ in case of~\eqref{sq}-\eqref{rq}, and consider the linearized evolution equation around the DFE solution, that is,
\be{eq:lin}
\frac{d\mathbf x}{dt}=\(\mathrm F-\mathrm V\)\,\mathbf x
\ee
where $\mathrm F$ and $\mathrm V$ respectively denote the matrices associated with the rate of new infections
and the rates of transfer between compartments, \emph{i.e.},
\be{FV}
\mathrm F=\(\begin{array}{cc}0&\frac\beta q\\ 0&0\end{array}\)\quad\mbox{and}\quad\mathrm V=\(\begin{array}{cc}-\,\alpha&0\\ \alpha&-\,\gamma\end{array}\)
\ee
See~\cite{van_den_Driessche_2002} for details. According to~\cite[Lemma 1]{van_den_Driessche_2002}, we observe that the matrix $\mathrm F$ is non-negative and the matrix~$\mathrm V$ is non-singular. In this framework, the \emph{basic reproduction ratio} is defined as the largest eigenvalue of $\mathrm F\,\mathrm V^{-1}$. It is an elementary computation to check that 
\be{mat:spectralradius}
\mathrm F\,\mathrm V^{-1}=\(\begin{array}{cc}\frac\beta{q\,\gamma}&\frac\beta{q\,\gamma}\\ 0&0\end{array}\)
\ee
has two eigenvalues, $0$ and $\frac\beta{q\,\gamma}=\mathcal R_0^{(1)}(q)$. This proves that $\mathcal R_0^{(1)}(q)$ is the basic reproduction ratio, as defined by the method of the next generation matrix. In this framework, it is known from~\cite[Theorem~2]{van_den_Driessche_2002} that the DFE solution is stable if $\mathcal R_0^{(1)}(q)<1$ and unstable if $\mathcal R_0^{(1)}(q)>1$.

\section{A heterogenous model of social distancing}\label{Sec:SEIR2G}

\subsection{A simple model with two groups}

Let us consider a population divided in two groups indexed by $k=1$, $2$, in which the \emph{Susceptible} individuals have a factor of reduction of social interactions $q_k$ which differ in the two groups. We shall assume that each of these groups gather a fixed fraction of the population $p_k$ with $p_2=p$ small and $p_1=1-p$. While the group corresponding to $k=1$ observes a lockdown and has a factor $q_k>1$, we are interested in the situation in which the other group has no reduction of social interactions: $q_2=1$, or eventually has more social interactions than average before lockdown, corresponding to some $q_2<1$. The typical example is the case of health workers in a period of epidemic disease or supermarket cashiers, who have contacts with a much larger number of people than an average individual. It is of course very difficult to estimate $q_2$ and one should take into account the efficiency of barrier procedures. Instead of trying to make rough guesses for the value of $q_2$, we will vary it in order to see what is the impact on the solutions.

\medskip With a straightforward notation, let us split the population of \emph{Susceptible} individuals in two groups
\[
\s=\s_1+\s_2
\]
and consider for $k=1$, $2$ the system
\be{Systq2}
\frac{\s_k'}{\s_k}=-\,\beta_k\,\i\quad\mbox{with}\quad\beta_k=\frac\beta{q_k}\,,\quad\e'=\(\beta_1\,\s_1+\beta_2\,\s_2\)\i-\,\alpha\,\e\,,\quad\i'=\alpha\,\e-\,\gamma\,\i\,,\quad\r'=\gamma\,\i\,.
\ee
There are multiple possible variants and it would make sense, for instance, to distinguish $\i_1$ and $\i_2$ in the above equations, with detailed contamination rules. The above system has striking properties. It is for instance straightforward to see that the linearized system around the DFE solution, \emph{i.e.}, the matrix $\mathcal M(1,0)$, has a largest eigenvalue given by
\be{lambda}
\lambda(\alpha,\beta,\gamma)=\frac12\(\sqrt{(\alpha-\gamma)^2+4\,\alpha\,\gamma\,\mathcal R_0^{(2)}(q_1,q_2,p)}-\alpha-\gamma\)\,,
\ee
where the basic reproduction ratio $\mathcal R$, as defined in the method of the \emph{next generation matrix}, is given by
\be{eq:formulamean2R0}
\mathcal R_0^{(2)}(q_1,q_2,p)=\frac{(1-p)\,\beta_1+p\,\beta_2}\gamma\,.
\ee
One can indeed apply the method of Section~\ref{sec:roliterature} and observe that the only change lies in the matrix $\mathrm F$, where the coefficient $\beta/q$ has to be replaced by $(1-p)\,\beta_1+p\,\beta_2$, which establishes~\eqref{eq:formulamean2R0}. It is easy to deduce from~\eqref{lambda} that $\lambda(\alpha,\beta,\gamma)$ is positive if and only if $\mathcal R_0^{(2)}(q_1,q_2,p)>1$.

\medskip Modeling heterogeneous mixing in infectious disease dynamics when the population is subdivided by characteristics other than those that are disease-related, such as risk status or age, is not new. This has been considered for instance from the dynamical point of view in~\cite{Lachiany_2016} or in the case of sexually transmitted diseases and particularly in HIV/AIDS models, with groups that are not all defined by disease related properties. In this perspective, \emph{contact matrices} have been considered, which involve a detailed analysis of the transmission mechanisms. We can refer to~\cite{May_1984,hethcote1987epidemiological,Hyman_1988,Jacquez_1988,Adler_1992} for various considerations in this direction and to~\cite[Section~3]{zbMATH00898403} for a discussion of the \emph{homogeneous mixing fallacy} in the application to successful vaccination policies. The present paper ignores a number of issues like symmetry in transmission between groups and density-dependent transmission questions in order to focus on simple qualitative questions: any serious study with quantitative goals should of course address these issues with care: see for instance~\cite{Blackwood_2018} for a warning. For sake of simplicity, we have chosen to consider that the origin of the infected individuals (group $1$ or $2$) plays no role in the transmission. This has the simple consequence that the basic reproduction ratio $\mathcal R$ is in the end exactly the \emph{average of the ratios independently computed for each group}, as shown by~\eqref{eq:formulamean2R0}. In the regime corresponding to $p$ much smaller than $1-p$ and whatever the details are, our model is anyway good enough to show that what matters concerning the basic reproduction ratio is the average of the ratios\footnote{~For instance, it is not the average of the factors of reduction of social interactions which matters, as it is usually done when considering homogeneous populations represented as a single group.}. However, when one considers the epidemic peak and the large time asymptotics, the message is not only the average and this is what we explain next.

\subsection{Conservations and large time asymptotics}

System~\eqref{Systq2} inherits of the properties of the standard SEIR model. The conservation of mass
\be{seir2:mass}
\s_1(t)+\s_2(t)+\e(t)+\i(t)+\r(t)=1\quad\forall\,t\ge0
\ee
guarantees that all quantities are bounded by $1$ as long as they are nonnegative. For a solution of~\eqref{Systq2}, let us observe that
\be{seir2:conservation}
\frac d{dt}\(\log\s_k\)=-\,\beta_k\,\i=-\,\frac{\beta_k}\gamma\,\r'\,,
\ee
so that
\be{s_k(t)}
\s_k(t)=\s_k^0\,e^{-\frac{\beta_k\,\r(t)}\gamma}\quad\mbox{with}\quad\s_k^0=\s_k(0)\,e^\frac{\beta_k\,\r(0)}\gamma\,.
\ee
We shall moreover assume that
\be{s_k(0)}
\s_k(0)=p_k\,\s(0)
\ee
as the population can be considered, in the initial phase of the outbreak (for $t\le0$, in our setting), as a single group. This point could be reconsidered and studied as in~\cite{bacaer:hal-02509142} if one is interested in the dynamics of the epidemics, but it has a no significant impact on the stable equilibrium in the uncontrolled case. An analysis of the trajectories of~\eqref{Systq2} as in Section~\ref{Sec:Equilibrium1} shows that solutions globally exist and that there is a unique stable attractor $(\s_1^\star,\s_2^\star,0,0,\r^\star)$. We deduce from the conservation of mass
\be{seir2:massasymp}
\s_1^\star+\s_2^\star+\r=1
\ee
that the stable stationary solution is given as the unique solution with $\r^\star=r>\r(0)$ of
\be{Eq:Equilibrium2}
1=\sum_{k=1,2}\s_k^0\,e^{-\frac{\beta_k\,r}\gamma}+r\,,\quad\s_k^\star=\s_k^0\,e^{-\frac{\beta_k\,r}\gamma}\,,\quad k=1,\,2\,.
\ee
Under the condition that $\big(\e_k(0),\i_k(0)\big)\neq(0,0)$, we have that
\be{seir2:solnlim}
\lim_{t\to+\infty}\big(\s_1(t),\s_2(t),\e(t),\i(t),\r(t)\big)=(\s_1^\star,\s_2^\star,0,0,\r^\star)\,.
\ee

\subsection{Stable equilibrium and phase transition}\label{Sec:Equilibrium2}

The same discussion as in Section~\ref{Sec:Equilibrium1} can be done. To fix ideas, let us assume for instance, as a simplifying assumption, that $\s(0)=1-\varepsilon$, $\e(0)+\i(0)+\r(0)=\varepsilon$, and $\r_0=\vartheta\,\varepsilon$ for some given $\vartheta\in[0,1)$, so that $\s_k^0=\s_k(0)=p_k\,(1-\varepsilon)$. The equation for the equilibrium~\eqref{Eq:Equilibrium2} can be rewritten as
\be{Eq:Equilibrium2bis}
(1-\varepsilon)\,\((1-p)\,e^{-\frac{\beta_1\,(\r-\vartheta\,\varepsilon)}\gamma}+p\,e^{-\frac{\lambda(\alpha\,(\r-\vartheta\,\varepsilon)}\gamma}\)+\r=1
\ee
and the Taylor expansion 
\be{Taylor}
(1-\varepsilon)\,\Big(1-\mathcal R_0^{(2)}(q_1,q_2,p)\,\big(1+o(\varepsilon)\big)\,\r\Big)+\r=1+o(\r)
\ee
for $\r>0$, small, shows that there is a solution with $\r>0$, small, as $\varepsilon\to0_+$ if and only if $\mathcal R_0^{(2)}(q_1,q_2,p)<1$, thus giving a solution of order $\varepsilon$, which turns out to be the unique solution. Otherwise, the only positive solution of~\eqref{Eq:Equilibrium2bis} corresponds to some $\r>>\varepsilon$ and we have exactly the same phase transition as in Section~\ref{Sec:Equilibrium1}, with $\mathcal R_0^{(2)}(q_1,q_2,p)$ playing the role of an order parameter. When the disease spreads, we find that $\r^\star$ is of the order of the solution of~\eqref{Eq:Equilibrium2bis} with $\varepsilon=0$, that is, of the solution $\r=r$ of
\be{seir2:Eq}
(1-p)\,e^{-\frac{\beta_1\,r}\gamma}+p\,e^{-\frac{\beta_2\,r}\gamma}+r=1\,.
\ee
By convexity, we know that
\be{seir2:Jensen}
(1-p)\,e^{-\frac{\beta\,r}{\gamma\,q_1}}+p\,e^{-\frac{\beta\,r}{\gamma\,q_2}}\ge e^{-\mathcal R_0^{(2)}(q_1,q_2,p)\,r}\,.
\ee
If we choose $q$ such that $\mathcal R_0^{(1)}(q)=\mathcal R_0^{(2)}(q_1,q_2,p)$, which means $\frac1q=\frac{1-p}{q_1}+\frac p{q_2}$, then it is clear that the solution of
\be{seir2:epidemicsize}
(1-p)\,e^{-\frac{\beta\,r}{\gamma\,q}}+r=1
\ee
is larger than $\r^\star$. In other words, the \emph{epidemic size} is reduced if we replace~\eqref{sq}-\eqref{rq} by~\eqref{Systq2} with $q$, $q_1$ and $q_2$ as above.

\section{General heterogeneous distribution}\label{Sec:Theoretical}

In this section we extend the previous framework to the situation of an arbitrary number of different population classes. Each class, or group, contains the individuals that share a given value of the transmission rate $\beta/q$. More precisely we consider a probability space $(\Omega, \mathcal{F}, \Proba)$ and $\mathfrak{S}(\omega,t)$, $\mathfrak{B(\omega)}$ random variables on this space designating respectively the state of a random individual $\omega \in \Omega$ and its $\beta/q$ parameter. We denote by $\E[\cdot]$ the average operator.

For instance the case of Section~\ref{Sec:SEIR2G} corresponds to the situation when $\mathfrak{B}$ has only two values, $\beta_1=\beta / q_1$ and $\beta_2=\beta / q_2 $, and $\s_k(t) = \Proba \left[ \mathfrak{S}(t)= \mbox{\emph{``Susceptible''}}, \mathfrak{B} = \beta/ q_k \right]$, with $k=1,2$. In order to keep notation simple, we will suppose in the following that the conditional law $\Proba[\mathfrak{B}| \mathfrak{S}(0)= \mbox{\emph{``Susceptible''}}\,]$ of $\mathfrak{B}$ relative to being susceptible (at $t=0$) only takes a finite number of values $\beta_1, ..., \beta_K$, \emph{i.e.}, is of the form $\sum_{k=1}^K p_k\,\delta_{\beta_k}$, with $p_k=\Proba[\mathfrak{B}= \beta_k| \mathfrak{S}(0)= \mbox{\emph{``Susceptible''}}\,] $, but all results given below extend to the general case. Note that $\sum_{k=1}^K p_k = 1$. We adopt the notation $\mathfrak{R} = \mathfrak{B}/ \gamma$, $R_k = \beta_k / \gamma$ for any $k=1,...,K$ and define 
\be{def:skeir}
\begin{array}{lll}
&\s_k(t)=\Proba\left[\mathfrak{S}(t)=\mbox{\emph{``Susceptible''}},\,\mathfrak{B}=\beta_k\right]\,,\quad
&\e(t)=\Proba\left[\mathfrak{S}(t)=\mbox{\emph{``Exposed''}}\,\right]\,,\\[4pt]
\quad&\i(t)=\Proba\left[\mathfrak{S}(t)=\mbox{\emph{``Infected''}}\,\right]\,,
\quad&\r(t)=\Proba\left[\mathfrak{S}(t)=\mbox{\emph{``Recovered''}}\,\right]\,,
\end{array}
\ee
and
\be{def:ssk}
\s(t) = \Proba \left[ \mathfrak{S}(t)= \mbox{\emph{``Susceptible''}}\,\right]=\sum_{k=1}^K\s_k(t)\,.
\ee
With these notations $p_k = \s_k(0) / \s(0)$.
The evolution of $\(\s_k,\e,\i,\r\)$ is governed by the system of equations\footnote{~The mathematically rigorous formulation of the equations involves defining a continuous time Markov chain for any individual and the associated infection and recovering probabilities, see~\cite{laguzet_individual_2015} for details.}:
\begin{equation}
\frac{d\s_k}{dt} = - \beta_k\,\s_k\,\i\,,\quad
\frac{d\e}{dt}= \(\,\sum_{k=1}^K\beta_k\,\s_k\)\i - \alpha\,\e\,,\quad
\frac{d\i}{dt}=\alpha\,\e-\,\gamma\,\i\,,\quad
\frac{d\r}{dt}=\gamma\,\i\,.
\label{eq:generalsyst}
\end{equation}
Let us start by a simple observation.
\begin{lemma} The solution $\(\s_k,\e,\i,\r\)$ of~\eqref{eq:generalsyst} satisfies:
\begin{equation}
\frac{d\s}{dt} = -\,\bar{\beta}\,a(t)\,\s(t)\,\i(t)\,,\quad\frac{d\e}{dt} = \bar{\beta}\,a(t)\,\s(t)\,\i(t)\,- \alpha\,\e(t)\,,\quad\frac{d\i}{dt} = \alpha\,\i(t) - \gamma\,\i(t)\,,\quad\frac{d\r}{dt} = \gamma\,\i(t)\,,
\label{eq:alternativeseir}
\end{equation}
where $\bar{\beta}= \sum_k p_k\,\beta_k = \E [\mathfrak{B}\,|\,\mathfrak{S}(0)=\mbox{\emph{``Susceptible''}}\,]$ and $a(t)$ is a positive nonincreasing function of $t$ with $a(0)=1$.
\label{lemma:attentuation}
\end{lemma}
\begin{proof} System~\eqref{eq:alternativeseir} is satisfied 
with
\be{def:a}
a(t)=\frac{\sum_k\beta_k\,\s_k(t)}{\bar\beta\,\s(t)}\,.
\ee
By definition of $\bar{\beta}$, we know that $a(0)=1$. Using the equation for ${d\s_k}/{dt}$, we can compute the derivative of $a$ and conclude that
\be{decay:a}
\frac{da}{dt} = - \frac{\i(t)}{\bar\beta\,\s(t)^2}\({\textstyle\sum_k\s_k(t)\,\sum_k \beta_k^2\,\s_k(t)-\big( \sum_k \beta_k\,\s_k(t) \big)^2}\)\le0
\ee
by the Cauchy-Schwarz inequality.
\end{proof}
\begin{theorem}\label{Thm:Main} Suppose that $\r(0)=0$. Then~\eqref{eq:generalsyst} possesses the following properties:
\begin{enumerate}
\item \label{item:thmr0mean} The basic reproduction ratio is the average of the reproduction ratios, \emph{i.e.}, $\overline{\mathfrak{R}}=\bar{\beta}/\gamma$. In probabilistic notation, $\overline{\mathfrak{R}} = \E [\mathfrak{R}\,|\,\mathfrak{S}(0)=\mbox{\emph{``Susceptible''}}\,]$~\footnote{~Here we will only consider the value of $\overline{\mathfrak{R}}$ for the event $\{ \mathfrak{S}(0)=\mbox{\emph{``Susceptible''}} \}$. Although the notation designates a random variable, we will consider it, when there is no ambiguity, as a real number.}.
\item \label{item:thmpeak} If $\overline{\mathfrak R}>1$, the peak (defined as the maximum value attained by $\e+\i$) is smaller than the peak obtained with a Dirac mass distribution (\emph{i.e.}, having only one group) with the same basic reproduction ratio.
\item \label{item:thmsize} 
If $\overline{\mathfrak R}>1$, the total epidemic size $\zeta$ is the unique solution of
\begin{equation}
1- \zeta = \s(0)\;\E\left[e^{- \mathfrak{R}\,\zeta}\,|\,\mathfrak{S}(0)=\mbox{\emph{``Susceptible''}}\,\right].
\label{eq:zetageneral}
\end{equation}
For any distribution, $\zeta$ is smaller than the total epidemic size of a Dirac mass distribution having the same average.
\end{enumerate}
\end{theorem}
By imposing the condition $\r(0)=0$, what we have in mind as initial datum is a perturbation of the DFE solution, or a solution with initial values for which the \emph{infectious compartments} are non-empty at $t=0$, \emph{i.e.}, $\e(0)+\i(0)>0$, eventually small (but this is not even mandatory). As we shall see,~\eqref{eq:zetageneral} follows by convexity as in Section~\ref{Sec:Equilibrium2}.

\begin{proof} Property~\eqref{item:thmr0mean} is obtained by linearization using the next generation method. The proof is the same as for~\eqref{eq:formulamean2R0} when $K=2$.

Next we prove Property~\eqref{item:thmpeak}. The peak is the value of $(\e+\i)$ at some time $t_p$ such that $(\e+\i)'(t_p)=0$, or, using Lemma~\ref{lemma:attentuation}, for the unique $t=t_p$ such that $\bar{\beta}\,a(t_p)\,\s(t_p) = \gamma$. In particular, since $a(t) \le 1$, we know that $\s(t_p) \ge 1/\,\overline{\mathfrak{R}}$, which is precisely the value of $\s$ at the peak for the homogeneous case $a = 1$. The value of the peak can be written as:
\begin{multline}
(\e+\i)(t_p)=(\e+\i)(0) + \int_0^{t_p}\(\bar{\beta}\,a(t)\,\s(t)\,\i(t) - \gamma\,\i(t)\)dt = (\e+\i)(0) + \int_0^{t_p}\(\frac{{(\overline{\mathfrak{R}})}^{-1}}{a(t)\,\s(t)}-1\)\frac{d\s}{dt}(t)\,dt \\ 
\le (\e+\i)(0) + \int_0^{t_p} {(\overline{\mathfrak{R}})}^{-1}\,\frac d{dt}\log\s(t)\,dt+\s(0)-\s(t_p) 
= (\e+\i)(0) + 
h(\s(t_p)) - h(\s(0))\,,
\end{multline}
where $h(x)={(\overline{\mathfrak{R}})}^{-1}\log x-x$. Since the function $x\mapsto h(x)$ is monotone decreasing on $[1/\,\overline{\mathfrak{R}},1]$, we obtain 
\be{Estim:peak}
(\e+\i)(t_p)\le (\e+\i)(0) + 
h(1/\,\overline{\mathfrak{R}}) - h(\s(0))
= 1-\frac{1}{\overline{\mathfrak R}}\left[1+\log 
\left( \overline{\mathfrak R}\,\s(0)\right) \,\right]\,.
\ee

Property~\eqref{item:thmsize} is obtained as follows. Denote by $\zeta$ the epidemic size, which satisfies the two identities: $\zeta = \r(\infty)= \r(\infty) - \r(0) = \s(0)+\e(0)+\i(0) - \sum_k \s_k(\infty)$. Thus $\zeta = \int_0^\infty \r'(t)\,dt = \gamma \int_0^\infty \i(t)\,dt$. For any $k \le K$: $s_k(\infty) = s_k(0)\,\exp\(-\int_0^\infty \beta_k\,\i(t)\,dt\) = \s(0)\, p_k\, e^{- \zeta\,\beta_k / \gamma}$. Combining the above relations we obtain~\eqref{eq:zetageneral}. In particular if the distribution is reduced to a Dirac mass centered at the mean value $\overline{\mathfrak{R}}$, the epidemic size, denoted by $\zeta_0$, satisfies: $\zeta_0 = 1- \s(0)\,e^{-\overline{\mathfrak{R}}\, \zeta_0}$ as expected. We use now the Jensen inequality for the convex function $x\mapsto e^{-x}$:
\begin{equation}
1- \zeta = \s(0)\,\E[e^{-R_\omega\,\zeta}\,|\,\mathfrak{S}(0)=\mbox{\emph{``Susceptible''}}\,] \ge\s(0)\,e^{-\overline{\mathfrak{R}}\,\zeta} = 
\s(0) \left( e^{-\overline{\mathfrak{R}}\,\zeta_0} \right)^{\zeta/\zeta_0}
 = \s(0) \left( \frac{1- \zeta_0}{\s(0)} \right)^{\zeta/\zeta_0}.
\end{equation}
We obtained thus $ \left(\frac{1- \zeta}{\s(0)} \right)^{1/\zeta} \ge \left(\frac{1- \zeta_0}{\s(0)} \right)^{1/\zeta_0}$, which implies, after recalling that $ \left(\frac{1- x}{\s(0)} \right)^{1/x} $ is a decreasing function on $]0,\s(0)[$ that $\zeta \le \zeta_0$, our conclusion.
\end{proof}

Note that, by assuming that $(\e+\i)(0)$ is small, it is clear that the peak can be made small if $\overline{\mathfrak R}-1>0$ is small.

\section{Numerical results}

\subsection{Choice of a set of parameters}

It is not the purpose of this paper to discuss the values of the parameters in a SEIR model of the pandemic of Covid-19 and we shall simply choose values that can be found in the literature, for a purely illustrative purpose. However, for the interested reader, we list our sources and some entries in the rapidly growing literature on the topic.

\medskip Concerning the initial data for the SEIR model, we shall assume that the size of the French population is $N=67.8\times10^6$. The following values correspond to the situation for the Covid-19 in France on March 15, 2020, according to~\cite{bacaer:hal-02509142} and data released on a daily basis by SPF at~\cite{SPF}. We shall take as initial data the values
\[
S(0)=N-\mathcal R_0\,R(0)\,,\quad E(0)=5970\,,\quad I(0)=1278\,,\quad R(0)=\frac{E(0)+I(0)}{\mathcal R_0-1}\approx 5450
\]
with $\mathcal R_0=2.33$. Here we provide the formulae used in~\cite{bacaer:hal-02509142} to infer the numbers, which are based on an asymptotic analysis of the SEIR model during the initial phase of the epidemics (without lockdown). Such values are important for controlled epidemics but play essentially no role if the disease is spreading in the population, which is the case under investigation. Note that it is very likely that the numbers will be revised in the future, to account for non well documented cases at the time this study was done. In any case, for numerical computations, we take the following initial data, as fractions of the total population
\be{InitialData}
\s(0)=0.99981\,,\quad\e(0)=8.81\times10^{-5}\,,\quad\i(0)=1.88\times10^{-5}\,,\quad\r(0)=8.04\times10^{-5}\,,
\ee
in all our examples. 

\medskip Now comes the issue of estimating the parameters $\beta$, $\gamma$ and $\alpha$ of the SEIR model, and the factor of reduction of social interactions $q$, at least as an average, in a socially homogeneous model and also for the majority of the population in the model with two groups. The methodology is out of the scope of this paper and we will not comment it. We refer to~\cite{magal2020.03.21.20040154} for a recent discussion. However, in order to fix the order of magnitude and give an idea of the uncertainties, let us review some numbers which recently appeared. 
An estimate based on the statistics of the known cases of Covid-19 at the beginning of the outbreak has been proposed in~\cite{bacaer:hal-02509142}, with the following values
\be{BacaerParameters}
\beta=2.33\,,\quad\alpha=0.25\,,\quad \gamma=1\,,\quad \mathcal R_0=2.33
\ee
before lockdown. By fitting the values for the first two weeks of lockdown, N.~Baca\"er came up with $q\approx1.7$ for an estimate of the factor of reduction of social interactions. These numbers are the ones used for our choice~\eqref{InitialData}. We learn from~\cite{inserm2020} that the parameters in the SEIR model can be estimated by: $\beta=2$, $\alpha=1/3.7\approx0.27$, $\gamma=1/1.5\approx0.67$ and $\mathcal R_0=3$ at regional level (Ile-de-France) and it was suggested that $q\approx2.94$ so that $\mathcal R_0/q=0.68$. This is consistent with the estimate $0.67$ of~\cite{Salje_2020}.

Before going further, let us list some limitations of our model and features which have to be taken into account in more realistic models. In~\cite{magal2020.03.21.20040154}, it is argued that \emph{Unreported} cases (U) should be taken into account in a S(E)IRU model for correctly accounting the Covid-19 outbreaks (also see~\cite{liu2020covid}). There is certainly an important point here, although we did not introduce it for sake of simplicity. Note that such models produce a very high~$\mathcal R_0$ (at least on the basis of the data used for fitting the curves) ranging from $4.45$ to $4.49$. In~\cite{EHESP}, the basic reproduction ratio is adjusted to $\mathcal R_0=2.8$. In~\cite{Salje_2020}, the authors find that the basic reproduction ratio, which was $\mathcal R_0\approx 3.41$ before lockdown, has been reduced to $0.52$, which corresponds to $q\approx6.6$. This is consistent with the factor $q=7$ found in~\cite{Roques_2020}. Of course important factors like risk status or age (see for instance~\cite{Gerasimov_2020,griette_2020age}) have been put in evidence, but they are not easy to study from the point of view of transmission rates. In our aproach, it is clear that we take oversimplifying assumptions, but this is probably the price to pay to prove results and put them numerically in evidence.

Recent papers have emphasized that the heterogeneity of the transmissions rates is important in the current pandemic and that it lowers herd immunity thresholds in various modeling frameworks. During the revision of this paper, we became aware of~\cite{Gomes_2020,Gerasimov_2020,Endo_2020}. Superspreaders and superspreading events are known to play an essential role in the propagation of Covid-19 according to~\cite{Adam_2020,Kucharski_2020,Prem_2020,Riou_2020,Althouse,H_bert_Dufresne_2020,Grossmann_2020}. Incorporating these considerations in compartmental models will anyway require further studies.

\medskip So far it remains difficult to choose a set of parameters, although one can hope that a better understanding of the dynamics of the pandemic will emerge out of the various studies that are currently done. What we learn after the end of the lockdown is still to be studied. Now comes a very empirical observation, in the framework of SEIR models, with parameters that are supposed to be constant in time, a crude assumption that is definitely not valid after the end of the lockdown. This empirical observation is the starting point for this paper. Situations with a basic reproduction ratio $\mathcal R_0<1$ exhibit an exponential decay of the number of cases while the data of~\cite{SPF} (at the beginning of the lockdown) were emblematic of a situation with $\mathcal R_0>1$. Several phases for a population under lockdown are illustrated in~\cite{lin2020analysis} by the computation of the \emph{effective reproduction ratio $\mathcal R_t$} (which is of the order of $\mathcal R_0$ as only a small proportion of the total population was concerned) in Wuhan, based on the reported daily Covid-19 infections:
when there was no intervention (before January 23, 2020):
$\mathcal R_t\approx 3.88$;
under lockdown with traffic ban and many confined at home (between January 23-February 1st, 2020):
$\mathcal R_t\approx 1.25$;
using \emph{centralized confinement} (after February 1st, 2020):
$\mathcal R_t\approx 0.32$.
Also see~\cite{Liu_2020} for a detailed analysis of the data in Wuhan based on a SIR type model. The curves in France and southern Europe during the first month of lockdown look more like to what happened in Wuhan by the end of january, with an $\mathcal R_0>1$, than to what happened later. This is what we intend to explain in France, at least partially, by social heterogeneity. For the sake of simplicity, we shall retain the values of~\eqref{InitialData}, and one of the reasons is that $q\approx1.7<\mathcal R_0$ is compatible with the above remarks, but our models of Sections~\ref{Sec:SEIR2G} and~\ref{Sec:Theoretical} show that there might be some subtleties when heterogeneities are taken into account.

\subsection{SEIR model with a single group}

We start by considering the standard SEIR model~\eqref{s}-\eqref{r} or its variant~\eqref{sq}-\eqref{rq}, where $q$ is the factor of reduction of social interactions. The basic reproduction ratio in the second system is $\mathcal R_0^{(1)}(q)=\mathcal R_0/q$, which allows us to reduce both problems to~\eqref{s}-\eqref{r} with various values of the basic reproduction ratio depending on the $q$ factor. Here we choose the initial data according to~\eqref{InitialData} and the set of parameters~\eqref{BacaerParameters}. See Fig.~\ref{f1-f2} for the epidemic curves, and Fig.~\ref{f3} for the epidemic size.

\setlength\unitlength{1cm}
\begin{figure}[ht]\begin{picture}(16,5)
\put(0,0){\includegraphics[width=8cm]{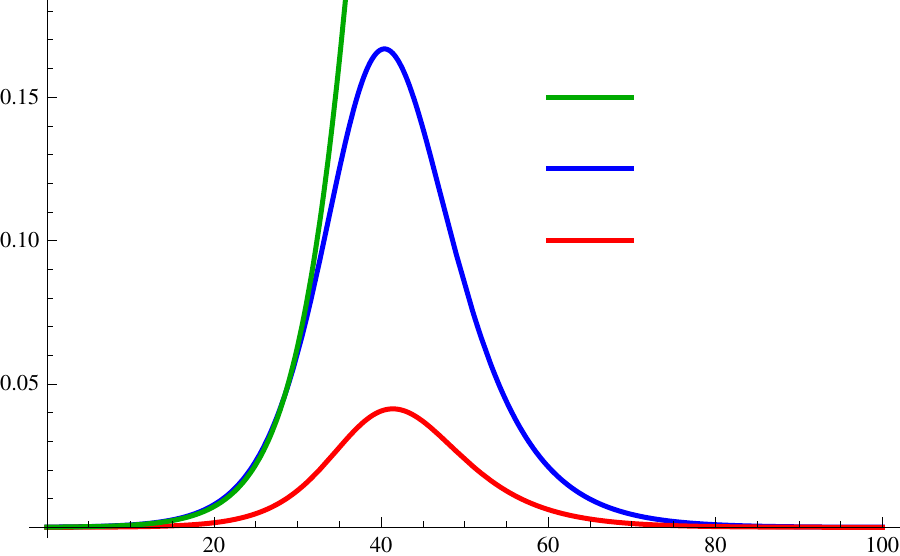}}
\put(8,0){\includegraphics[width=8cm]{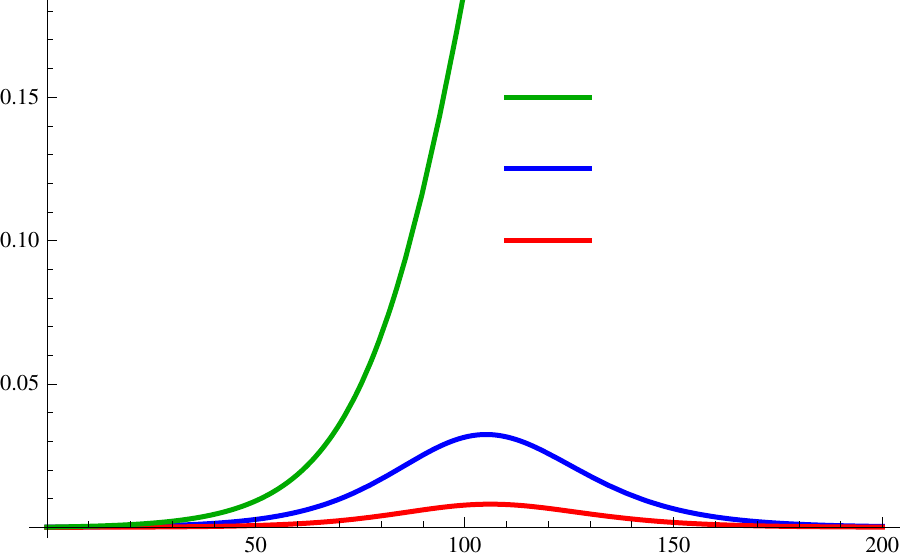}}
\put(5.75,4){\small\emph{Recovered}}
\put(5.75,3.33){\small\emph{Exposed}}
\put(5.75,2.67){\small\emph{Infected}}
\put(13.5,4){\small\emph{Recovered}}
\put(13.5,3.33){\small\emph{Exposed}}
\put(13.5,2.67){\small\emph{Infected}}
\put(7.5,0.5){$t$}
\put(15.5,0.5){$t$}
\put(1,4.5){$q=1$}
\put(1,3.83){$\mathcal R_0=2.33$}
\put(9,4.5){$q=1.7$}
\put(9,3.83){$\mathcal R_0^{(1)}=1.37$}
\end{picture}
\caption{\label{f1-f2} The peak of the outbreak in the SEIR model. The time $t$ is counted in days. The vertical axis represents the fraction of the population. The basic reproduction ratio is either $\mathcal R_0=2.33$ (left) or $1.37$ (right) corresponding to a reduction of social interactions by a factor $q=1.7$ as in~\cite{bacaer:hal-02509142}. This illustrates the \emph{flattening of the curves}.}
\end{figure}
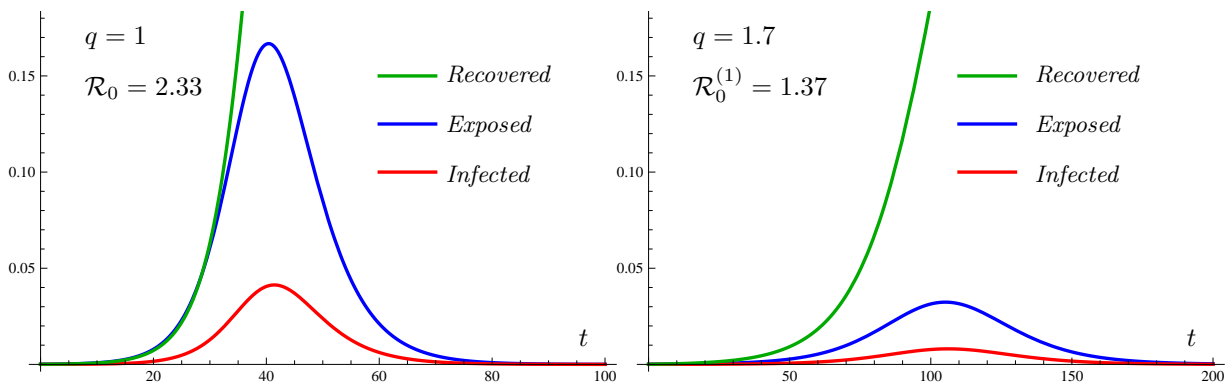

\begin{figure}[ht]
\begin{picture}(8,5)
\put(0,0){\includegraphics[width=8cm]{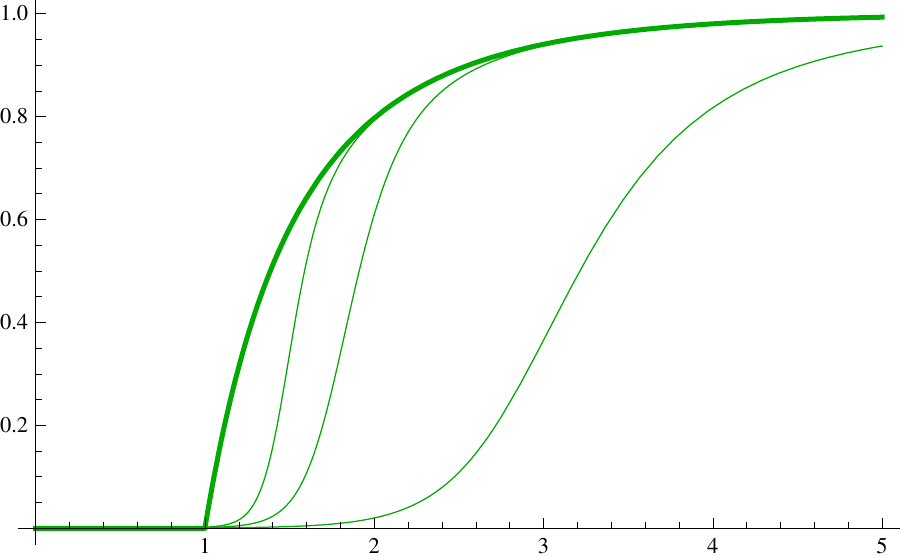}}
\put(2.75,4){$\zeta$}
\put(2.3,2){$t=90$}
\put(2.6,1.4){$t=60$}
\put(3.95,0.9){$t=30$}
\put(7.5,0.5){$\mathcal R_0$}
\end{picture}
\caption{\label{f3} The epidemic size $\zeta$ as a function of the basic reproduction ratio $\mathcal R_0$ in the simplest SEIR model corresponding to~\eqref{s}-\eqref{r} exhibits a clear phase transition at $\mathcal R_0=1$ (thick curve). The other curves represent $\r(t)$ taken for $t=30$, $60$, $90$, for which equilibrium is not yet achieved. In practice, varying $\mathcal R_0$ is achieved by acting on the $q$ factor in~\eqref{sq}-\eqref{rq}.}
\end{figure}
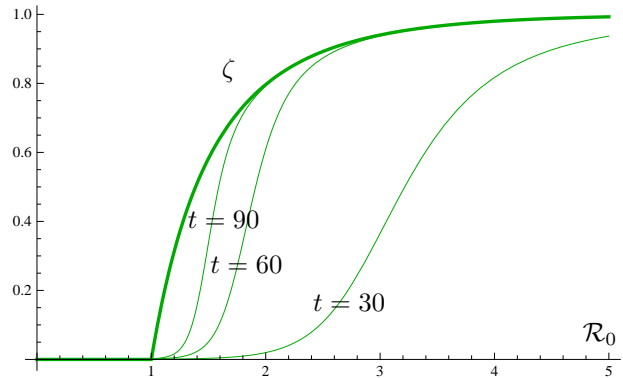

\subsection{SEIR model with two groups having different factors of reduction of social interactions}

One of the most disturbing results in~\cite{bacaer:hal-02509142} is that fitting the data of the cumulated number of cases by model~\eqref{sq}-\eqref{rq} shows that $q\approx1.7$, which is far below $\mathcal R_0=2.33$ and suggest that lockdown is inefficient for controlling the outbreak. However, we have seen that there is another possible model, as it is illustrated by Fig.~\ref{f4}, with two groups. From the figure, it is clearly impossible to distinguish between the two scenarios at the early stage of the outbreak.
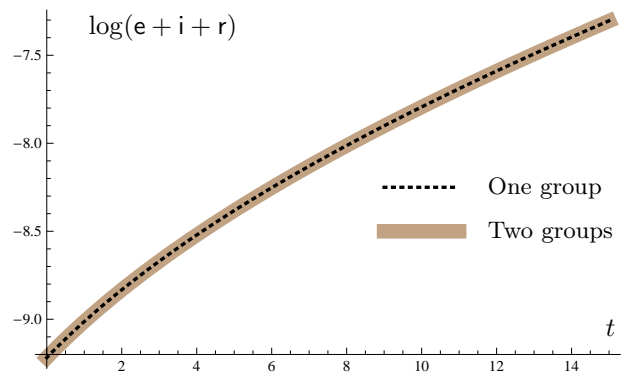
\begin{figure}[hb]\begin{picture}(8,5)
\put(0,0){\includegraphics[width=8cm]{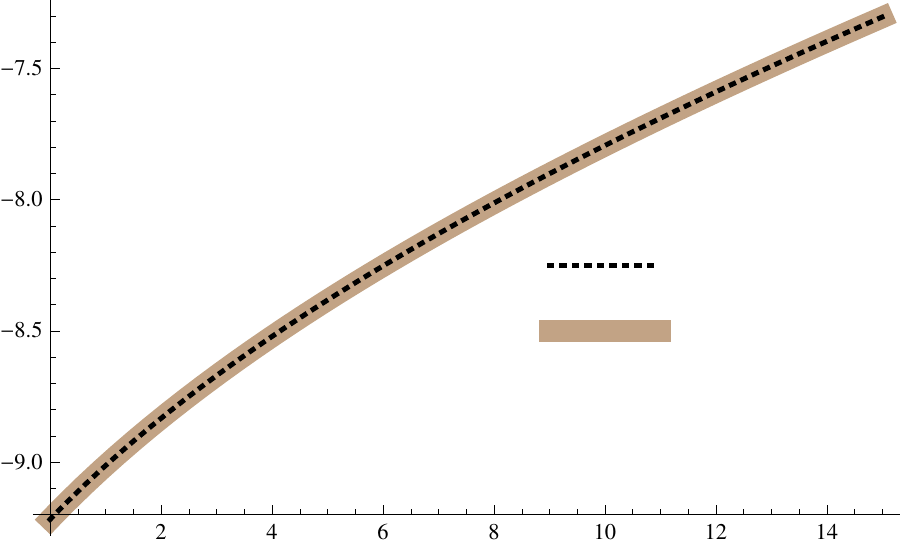}}
\put(7.8,0.5){$t$}
\put(1,4.5){$\log(\e+\i+\r)$}
\put(6.25,2.38){\small{One group}}
\put(6.25,1.8){\small{Two groups}}
\end{picture}
\caption{\label{f4} Plot of $t\mapsto\e(t)+\i(t)+\r(t)$ for a solution of~\eqref{sq}-\eqref{rq} with $q=1.7$ (black dotted line) \emph{versus} a plot of $t\mapsto\e(t)+\i(t)+\r(t)$ for a solution of~\eqref{Systq2} in logarithmic scale with $p=0.02$, $q_1=2.35$, and $q_2=0.117$ (brown). In both cases the basic reproduction ratio is $1.37$.}
\end{figure}

With two groups arises the question of choosing the initial data. As a crude and simplifying assumption, we make the choice to consider that before lockdown there was a single population and that the two categories of the population had the same exposure to the disease. We can then use~\eqref{InitialData} as initial condition and take $\s_1(0)=(1-p)\,\s(0)$ and $\s_2(0)=p\,\s(0)$. With this choice, we recover the results of~\eqref{sq}-\eqref{rq} if $q_1=q_2$ and $\s=\s_1+\s_2$. In order to fix ideas, we also make the arbitrary choice of choosing $q_1=2.35$ so that the epidemic disease would extinguish by itself after affecting $0.81\%$ of the population if $q_2=q_1$. We can illustrate the role of the two parameters $p$ and $q_2$ by showing that they completely change the picture and bring us back to a regime with an epidemic size corresponding to some inefficient lockdown, however with lower epidemic peak and size: see Fig.~\ref{f5-f6}, and Fig.~\ref{f7} for the epidemic size.
\begin{figure}[hb]\begin{picture}(16,5)
\put(0,0){\includegraphics[width=8cm]{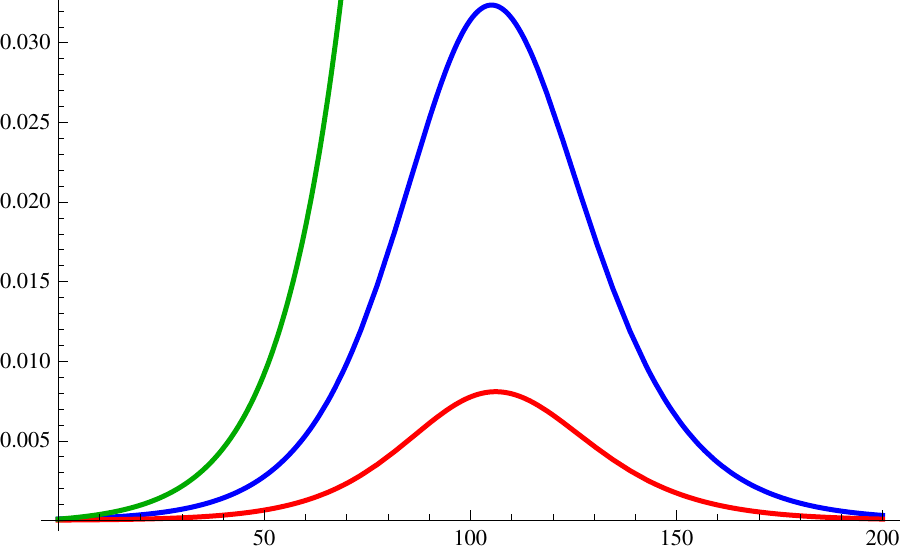}}
\put(8,0){\includegraphics[width=8cm]{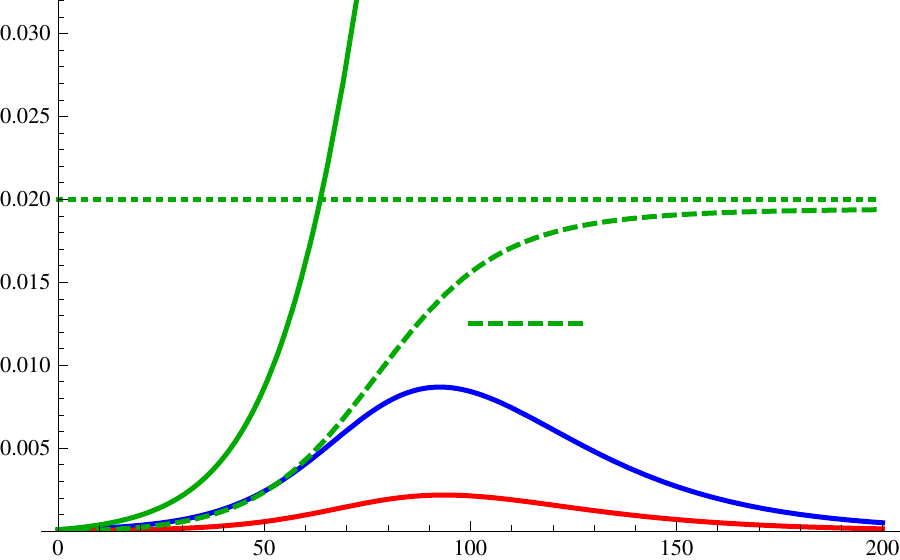}}
\put(5.5,4.8){\small{One group}}
\put(5.5,4.2){$\mathcal R_0^{(1)}=1.37$}
\put(13.5,4.8){\small{Two groups}}
\put(13.5,4.2){$\mathcal R_0^{(2)}=1.37$}
\put(13.25,2){\small\emph{Recovered, group 2}}
\end{picture}
\caption{\label{f5-f6} Model with a single group (left) and $q=1.7$ corresponding to a basic reproduction ratio of $1.37$, and two groups (right) with $q_1=2.35$, $q_2=0.117$, and $p=0.02$ as in Fig.~\ref{f4}, with same basic reproduction ratio. Note that the figure on the left is the same as in Fig.~\ref{f1-f2} (right), on a different scale. The straight dotted line is the level $p$. In the case with two groups, note that almost all individuals of the second group get infected during the propagation of the disease.}
\end{figure}
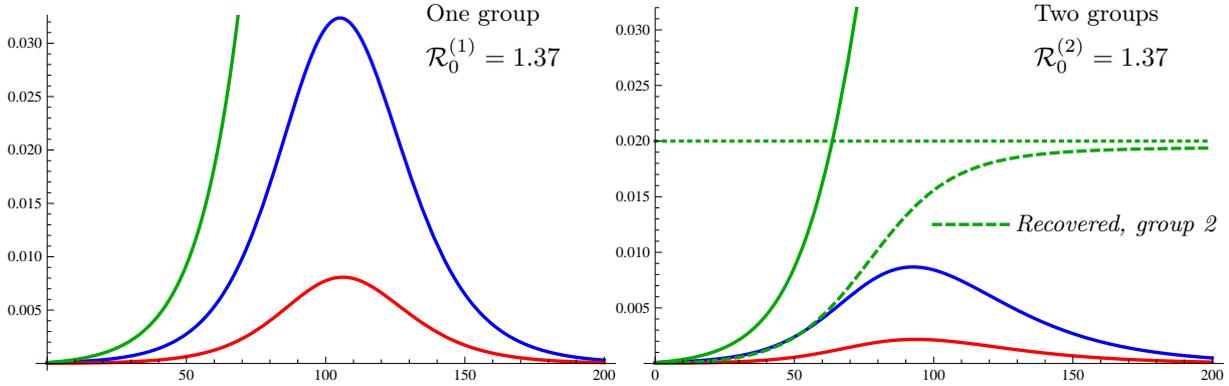

\begin{figure}[ht]\begin{picture}(8,5)
\put(0,0){\includegraphics[width=8cm]{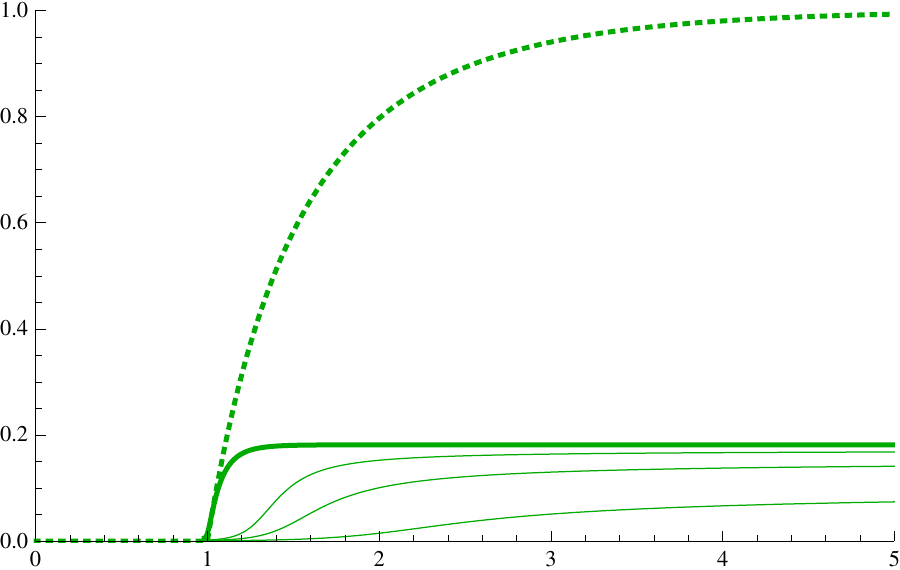}}
\put(5.5,4.4){$\zeta$\small{, one group}}
\put(5.5,1.5){$\zeta$\small{, two groups}}
\put(8.25,0.25){$\mathcal R_0$}
\end{picture}
\caption{\label{f7} The epidemic size in a population with one group (dotted curve, same as in Fig.~\ref{f3}) or two groups (plain, thick curve) with $q_1=2.4>\mathcal R_0=2.33$ and $p=0.02$ as a function of the basic reproduction ratio (obtained by varying $q$ in the first case and $q_2$ in the second case). The other lines correspond to simulations of $\r(t)$ for $t=30$, $60$, $90$ in the model with two groups. With two groups, we recover a phase transition as in~Fig.~\ref{f3}, corresponding to $q_2\approx0.96$.}
\end{figure}
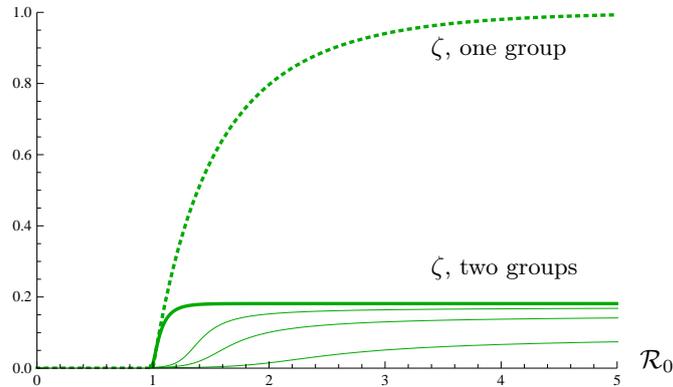

\section{Discussion}\label{Sec:Discussion}

In the SEIR model with a basic reproduction ratio $\mathcal R_0>1$, the stable equilibrium does not depends much on the initial data: it is almost entirely determined by $\mathcal R_0$. The goal of lockdown is to decrease $\mathcal R_0$ by dividing it by a reduction factor for social interactions, $q$. With $\mathcal R_0/q<1$, the disease is under control. With $\mathcal R_0/q>1$, but close to $1$, the epidemic disease spreads, the final state depends little on the initial data, but the curve is flattened: qualitatively, this is the scenario that one can observe in the Covid-19 outbreak in France under lockdown.

Measuring the $q$ factor is difficult. The aim of this article is to show that the crucial information cannot be reduced to the knowledge of an average factor $q$: if the population is divided into two groups, with a group for which $q>\mathcal R_0$ (the majority) and another group (the minority) that keeps a small $q$ factor, the disease may continue to spread. If the $q$ factor of the majority is larger than $\mathcal R_0$ but close to $\mathcal R_0$, so that $1-\mathcal R_0/q>0$ is small, the impact of the minority becomes extremely important as it eventually triggers the nonlinearity. However, the equilibrium asymptotic state in a two-group model is not the same as when considering a single group with an averaged basic reproduction ratio. The dynamics of the outbreak, for instance the height of the epidemic peak, is also changed. A two-group model is of course extremely simplistic, but shows the importance of understanding the distribution of the $q$ factors in a population.

Our observations are not limited to a population divided into two groups. In our model of heterogeneous social interactions, with a whole distribution of $q$ factors, we have shown that an average of $q$ is not relevant. While the basic reproduction ratio behaves as a plain average across the homogeneous $q$~categories, neither the peak nor the total epidemic size do the same. In particular the presence of heterogeneity is beneficial for both the peak and total epidemic size. Or, put it otherwise, a model with only one group and fitting the observed data in the initial phase of the outbreak will be more pessimistic concerning the epidemic outcomes than a heterogeneous model; this is even more true after lockdown when social distancing measures have been enforced, the lockdown being by its nature a creator of heterogeneity. In terms of public health, this also underlines the importance of targeting prevention measures on individuals with a high level of social interactions.

\begin{acknowledgement} This work has been partially supported by the Project EFI (ANR-17-CE40-0030) of the French National Research Agency (ANR). The authors thank Nicolas Baca\"er for a stimulating discussion which was the starting point of this research project, Gilles Z\'erah for pointing them an important reference, and an anonymous referee for pointing them a missing argument and suggesting several improvements. The authors also thank the \href{https://insmidirect.math.cnrs.fr/spip.php?article3647}{MODCOV19} platform for encouragements and collective effort. \end{acknowledgement}

\end{document}